\documentclass{amsart}

\usepackage{amsthm}
\usepackage{amsmath} \usepackage{amsfonts}
\usepackage{amssymb} \usepackage{latexsym} \usepackage{enumerate}
\usepackage{mathtools}
\mathtoolsset{showonlyrefs}
\usepackage{todonotes}
\usepackage{wrapfig}

\usepackage{mathrsfs}
\usepackage[numbers]{natbib}

\numberwithin{equation}{section}
\newtheorem{thm}{Theorem}[section]

\newtheorem{lem}[thm]{Lemma}
\newtheorem{rem}[thm]{Remark}
\newtheorem{cor}[thm]{Corollary}
\newtheorem{prop}[thm]{Proposition}
\newtheorem{asm}[thm]{Assumption}

\newcommand{\kkernel}{k}
\begin{document}
\title{What if we knew what the future brings?\\ Optimal investment for
  a frontrunner with price impact} \thanks{P. Bank is supported in
part by the
GIF Grant 1489-304.6/2019.\\
Y. Dolinsky is supported in part by the GIF Grant 1489-304.6/2019 and
the ISF grant
230/21.\\
M. Rásonyi thanks for the support of the “Lendület” grant LP 2015-6 of
the Hungarian Academy of Sciences.  }

   \author{Peter Bank \address{
 Department of Mathematics, TU Berlin. \\
 e.mail: bank@math.tu-berlin.de }}${}$\\
\author{Yan Dolinsky \address{
 Department of Statistics, Hebrew University of Jerusalem.\\
 e.mail: yan.dolinsky@mail.huji.ac.il}}
 \author{Miklós Rásonyi \address{
 Alfr\'ed Rényi Institute of Mathematics and E\"otv\"os Lor\'and University, Budapest.\\
 e.mail: rasonyi@renyi.hu}}
  \date{\today}
\maketitle
\begin{abstract}
In this paper we study optimal investment when the investor can peek
some time units into the future, but cannot fully take advantage of
this knowledge because of quadratic transaction costs. In the
Bachelier setting with exponential utility, we give an explicit
solution to this control problem with intrinsically
infinite-dimensional memory. This is made possible by solving the dual
problem where we make use of the theory of Gaussian Volterra integral equations.
\end{abstract}
\begin{description}
\item[Mathematical Subject Classification (2010)] 91G10, 91B16
\item[Keywords] Gaussian Volterra integral equation, inside
  information, price impact, exponential utility, optimal investment
\end{description}

 \maketitle \markboth{}{}
\renewcommand{\theequation}{\arabic{section}.\arabic{equation}}
\pagenumbering{arabic}

\section{Introduction}\label{sec:intro}

Optimal investment is a tremendously rich source of mathematical
challenges in stochastic control theory. The key driver in this
problem is the tradeoff between risk and return. Thus, information on
the investment opportunities is playing a role which is as important
for the mathematical theory as it is in practice where investors go at
great lengths to secure even the slightest advantage in knowledge.  So
it is no wonder that insider information has been widely studied in
the literature; see, for instance,
\cite{PK:96,APS:98,I:03,ABS:03,ADI:06} where an investor obtains extra
information about the stock price evolution at some fixed point in
time. By contrast to these studies, the present paper takes a more
dynamic view on information gathering and affords the investor the
opportunity to \emph{continually} peek $\Delta$ units of time into the
future.  Closest to such an investor in reality may be high-frequency
traders (``frontrunners'') that get access to order flow information
earlier or are able process it faster than their competition. To the
best of our knowledge, this paper is the first continuous-time
stochastic control paper with such a feature, apart from the optimal
stopping problem of \cite{BZ:16}.

Of course, perfect knowledge about future stock prices easily lets
optimal investment problems degenerate and so it is of great interest
to understand how market mechanisms may curb an investor's ability to
take advantage of this extra information. A most satisfactory approach
from an economic point of view is the equilibrium approach due to Kyle
\cite{Kyle:85} where the insider knows the terminal stock price right
from the start and internalizes the impact of her orders on market
prices. Generalizations of this approach are challenging; see
\cite{BackBaruch:04}, \cite{Cetin:18}, \cite{BackEtal:21} and the
references therein.  For dynamic information advantages in this
context, we refer to \cite{CampiCetinDanilova:11,
  CampiCetinDanilova:13} who consider an insider receiving a dynamic
signal on, respectively, the terminal asset price or the traded firm's
default time. These models, however, do not get close to addressing
the intrinsically infinite-dimensional information structure of our
peek-ahead setting. Fortunately, also the much simpler market impact
model of \cite{AlmgrenChriss:01} that just imposes quadratic
transaction costs for the investor turns out to be sufficient friction
to make the optimal investment problem viable. In an insider model
where additional information is obtained just once,
\cite{AnkirchnerEtAl:2016} use such a friction for optimal portfolio
liquidation. A combination between Kyle's equilibrium setting and
quadratic price impact costs is solved in \cite{Donnelly:21}. With
peek-ahead information as in the present paper, \cite{DZ:21} study
super--replication, albeit in a discretized version of the Bachelier
model.

It is in the continuous-time Bachelier model that the present paper
provides its main result, namely the explicit optimal investment
strategy for an exponential utility maximizer who knows about future
prices $\Delta$ time units before they materialize in the market, but
cannot freely take advantage of her extra knowledge due to quadratic
transaction costs. The optimal policy turns out to be a combination of
two trading incentives. On the one hand, there is the urge to trade
towards the optimal frictionless position given by the well-known
Merton ratio. On the other hand, there is the desire to take advantage
of the next stock price moves and this contributes to the optimal
turnover rate through an explicitly given average of stock prices over
the window of length $\Delta$ on which our investor has extra
information.

Due to its peek-ahead feature, our optimal control problem can be
viewed as a contribution to pathwise stochastic control. A closely
related work is \cite{BM:07} where the authors studied a hidden
stochastic volatility model with a controller who has full information
on the extra noise.  The theory of delayed or partial information also
shares the infinite-dimensional pathwise control issues we need to
address here; see the recent papers \cite{SZ:19,S:19} and the
references therein. Finally, our control theoretic setting is also
related to models discussed in the monograph \cite{Gozzi:17}.

Instead of dynamic programming (which would be challenging in this
infinite memory setting; cf. \cite{Gozzi:17}), our methodology is
based on duality. For the case of exponential utility and quadratic
transaction costs, this theory is developed with flexible information
flow in great generality in an essentially self-contained appendix. It
shows that the primal optimal control is determined by the conditional
expectation of the terminal stock price under the dual optimal
probability measure. For the Brownian framework that we focus on in
the main body of the paper, we derive a particularly convenient
representation of the dual target functional which leads to
\emph{deterministic} variational problems. These problems can be
solved explicitly, and results from the theory of Gaussian Volterra
integral equations (\cite{H:68,HH:93}) allow us both to construct the
solution to the dual problem and to compute the primal optimal
strategy. These Gaussian Volterra integral equations also occur in
\cite{Cheridito:01} albeit in the rather different context of (no)
arbitrage criteria in fractionally perturbed financial models.

In Section~\ref{sec:mainresult} we specify our model and formulate and
interpret our main result.  Section~\ref{sec:proof} contains the proof
of the main result and the appendix~\ref{appendix} presents the
duality results necessary for these developments.

\section{Problem Formulation and Main Result}\label{sec:mainresult}
We consider an investor who knows about market movements some time
before they happen, but cannot arbitrarily exploit them due to market
frictions. Specifically, apart from a riskless savings account bearing
zero interest (for simplicity), the investor has the opportunity to
trade in a risky asset with Bachelier price dynamics
\begin{align}
S_t=s_0+\mu t+\sigma W_t, \quad t \geq 0,
\end{align}
where $s_0 \in \mathbb R$ is the initial asset price,
$\mu\in\mathbb{R}$ is the constant drift, $\sigma>0$ is the
constant volatility and $W$ is a one-dimensional Brownian motion
supported on a complete probability space
$(\Omega, \mathcal{F}, \mathbb P)$. Rather than having access to just
the natural augmented filtration $(\mathcal{F}^S_t)_{t \geq 0}$ for
making investment decisions, we assume that our investor can peek
$\Delta \in [0,\infty)$ time units into the future, and so her information
flow is given by the filtration
$$
\mathcal G^{\Delta}_t:=\mathcal F^S_{t+\Delta}, \quad t\geq 0.
$$
\begin{rem}
  As suggested by an anonymous referee, one could more generally
  consider a non-decreasing time shift $\tau:[0,\infty) \to [0,\infty)$ with
  $\tau(t) \geq t$ to model time-varying ability to peek ahead. To
  keep the exposition here as simple as possible, we leave this
  extension of our model as a topic for future research.
\end{rem}
Taking advantage of the inside information is impeded by the
investor's adverse market impact. Following
\cite{AlmgrenChriss:01}, we model this impact in a temporary linear
form and, thus, when at time $t$ the investor turns over her position $\Phi_t$ at
the rate $\phi_t=\dot{\Phi}_t$ the execution price is $S_t + \frac{\Lambda}{2}
\phi_t$ for some constant $\Lambda >0$. As a result, the profits and
losses from trading are given by
\begin{align}\label{eq:pnl}
V^{\Phi_0,\phi}_T:=\Phi_0(S_T-S_0)+\int_{0}^T \phi_t(S_T-S_t)dt-\frac{\Lambda}{2} \int_{0}^T \phi^2_t dt,
\end{align}
where, for convenience, we assume that the investor marks to market her position
$\Phi_T = \Phi_0+\int_0^T \phi_t dt$ in the risky asset that she has acquired by time
$T>0$.

Fixing a time horizon $T>0$, the natural class of admissible strategies
is then
\begin{align}
\mathcal A^\Delta:=\left\{\phi=(\phi_t)_{t\in [0,T]}: \ \phi \text{ is } \
  \mathcal G^{\Delta}\text{-optional with } \int_{0}^T \phi^2_t dt<\infty
  \ \text{ a.s.}\right \}.
\end{align}
The investor’s preferences are described by an exponential utility function
$$u(x):=-\exp(-\alpha x), \quad x\in\mathbb R,$$
with constant absolute risk aversion parameter $\alpha>0$, and her goal
is thus to
\begin{align}\label{problem}
\text{Maximize } \mathbb E\left[u(V^{\Phi_0,\phi}_T)\right]=\mathbb
  E\left[-\exp\left(-\alpha V^{\Phi_0,\phi}_T\right)\right]\text{ over
  } {\phi\in\mathcal A^\Delta}.
\end{align}

The paper's main result is the following solution to this optimization
problem:
\begin{thm}\label{thm2.1}
  In the utility maximization problem~\eqref{problem}, the investor's
  optimal turnover rate $\hat{\phi}_t$ at time $t \in [0,T]$ depends
  on the risk-liquidity ratio
\begin{align}
  \label{eq:22}
  \rho:=\frac{\alpha \sigma^2}{\Lambda},
\end{align}
  on the position $\hat{\Phi}_t=\Phi_0+\int_0^t \hat{\phi}_sds$
  acquired so far and the privileged information on the next stock prices
  $(S_{t+s})_{s \in [0,\Delta]}$ in the feedback
  form
  \begin{align}\label{eq:feedback}
\hat{\phi}_t=&
               \frac{1}{\Lambda}\left(\bar{S}^\Delta_t-S_t\right)
                          +\frac{\Upsilon^\Delta(T-t)}{\Delta}
                            \left(\frac{\mu}{\alpha\sigma^2} -\hat{\Phi}_t\right),
  \end{align}
  where $\bar{S}^\Delta$ is the stock price average given by
  \begin{align}
    \label{eq:24}
    \bar{S}^\Delta_t := \left(1-\Upsilon^\Delta(T-t)\right)S_{(t+\Delta) \wedge T}
+\Upsilon^\Delta(T-t)\frac{1}{\Delta}
\int_0^{\Delta} S_{t+s}ds
  \end{align}
 with $\Upsilon^\Delta(\tau)=\Delta\sqrt{\rho}\tanh(\sqrt{\rho}(\tau-\Delta)^{+})/(1+\Delta\sqrt{\rho}\tanh(\sqrt{\rho}(\tau-\Delta)^+))$.
The maximal utility this policy generates is
  \begin{align}\label{eq:primalvalue}
  \max_{\phi \in \mathcal A^\Delta} \mathbb  E&\left[-\exp\left(-\alpha\sigma
                                                V^{\Phi_0,\phi}_T\right)\right] =\\\nonumber
    &-\exp\left(
      \frac{\alpha\Lambda\sqrt{\rho}}{2\coth\left(\sqrt{\rho}T\right)}{\left(\Phi_0-\frac{\mu}{\alpha
      \sigma^2}\right)^2}-\frac{1}{2}\frac{\mu^2}{\sigma^2}T\right)\\\nonumber&\qquad\cdot \exp\left(-\frac{1}{2}\int_{0}^T \frac{(s\wedge\Delta)\rho}{1+(s\wedge\Delta) \sqrt{\rho}\tanh\left(\sqrt{\rho}(T-s)\right)}ds\right).
  \end{align}
\end{thm}

Our feedback description~\eqref{eq:feedback} can be interpreted as
follows: First, without privileged information, i.e. for $\Delta=0$,
we have $\bar{S}^\Delta_t=S_t$ and, therefore, the first term
in~\eqref{eq:feedback} vanishes leaving us with the optimal policy
\begin{align}
         \hat{\phi}_t=\sqrt{\rho}\tanh(\sqrt{\rho}(T-t))
                            \left(\frac{\mu}{\alpha\sigma^2}
  -\hat\Phi_t\right), \quad t \in [0,T].
\end{align}
So the uninformed agent will trade towards the optimal position
$\mu/(\alpha\sigma^2)$ well known from the frictionless Merton
problem. Due to the impact costs, she does so with finite urgency
$\sqrt{\rho}\tanh(\sqrt{\rho}(T-t))$. With a long time to go, this
urgency is essentially $\sqrt{\rho}$ and thus dictated by the
risk/liquidity ration $\rho=\alpha \sigma^2/\Lambda$;
as $t$ approaches the time horizon $T$, the urgency vanishes because,
towards the end, position improvements have an ever shorter time to
yield risk premia but the investor still has to pay the same impact
costs that obtain at the start of trading.  \footnote{We will prove
  this result along the way to our main result with future knowledge
  $\Delta>0$. Let us note though that, for $\Delta=0$, a closely
  related result is obtained by dynamic programming techniques in
  \cite{Schied_2007} who, in contrast to our setting, impose a
  liquidation constraint $\hat{\Phi}_T=0$ and assume $\mu=0$.}

 For the informed agent, i.e. for $\Delta>0$, the desire to be close to
the Merton ratio persists, but the urgency reduces to
\begin{align}
\frac{\Upsilon^\Delta(T-t)}{\Delta}=\frac{\sqrt{\rho}\tanh(\sqrt{\rho}(T-t-\Delta)^{+})}{1+\Delta\sqrt{\rho}\tanh(\sqrt{\rho}(T-t-\Delta)^+)},
\end{align}
 leaving ``some air'' to take advantage of the knowledge on future
price movements. This is done by averaging out in~\eqref{eq:24} the
latest relevant stock price available to the investor,
$S_{(t+\Delta)\wedge T},$ with the mean stock price
$\frac{1}{\Delta} \int_0^{\Delta} S_{t+s}ds$ to be realized over the
next $\Delta$ time units in an effort to assess the earnings potential
over today's stock price~$S_t$. Put into relation with the impact costs $\Lambda$, this yields
the second contribution $(\bar{S}^\Delta_t-S_t)/\Lambda$ to the optimal turnover rate. The weight that this
assessment of earnings assigns to the average stock prices is given by
$\Upsilon^\Delta(T-t) \in [0,1]$; it is about
$\Delta \sqrt{\rho}/(1+\Delta\sqrt{\rho})$ when there is still a lot
of time to go, but vanishes completely as soon as $T-t\leq \Delta$,
i.e. as soon as full knowledge of stock price movements over the
relevant time span $[0,T]$ is attained. In this terminal regime also
the ambition to be close to the Merton ratio is wiped out and the
investor just chases the earning potential $S_T-S_t$ from the stock,
of course still in a tradeoff against the liquidity costs $\Lambda$;
this latter effect is also immediate from separate, pointwise
optimization over $\phi_t$ in the representation~\eqref{eq:pnl} of
profits and losses (which leads to $\phi_t^*=(S_T-S_t)/\Lambda$, $t
\in [0,T]$, an admissible
strategy as soon as $S_T$ becomes known).
\begin{figure}
\includegraphics[width=0.3\textwidth]{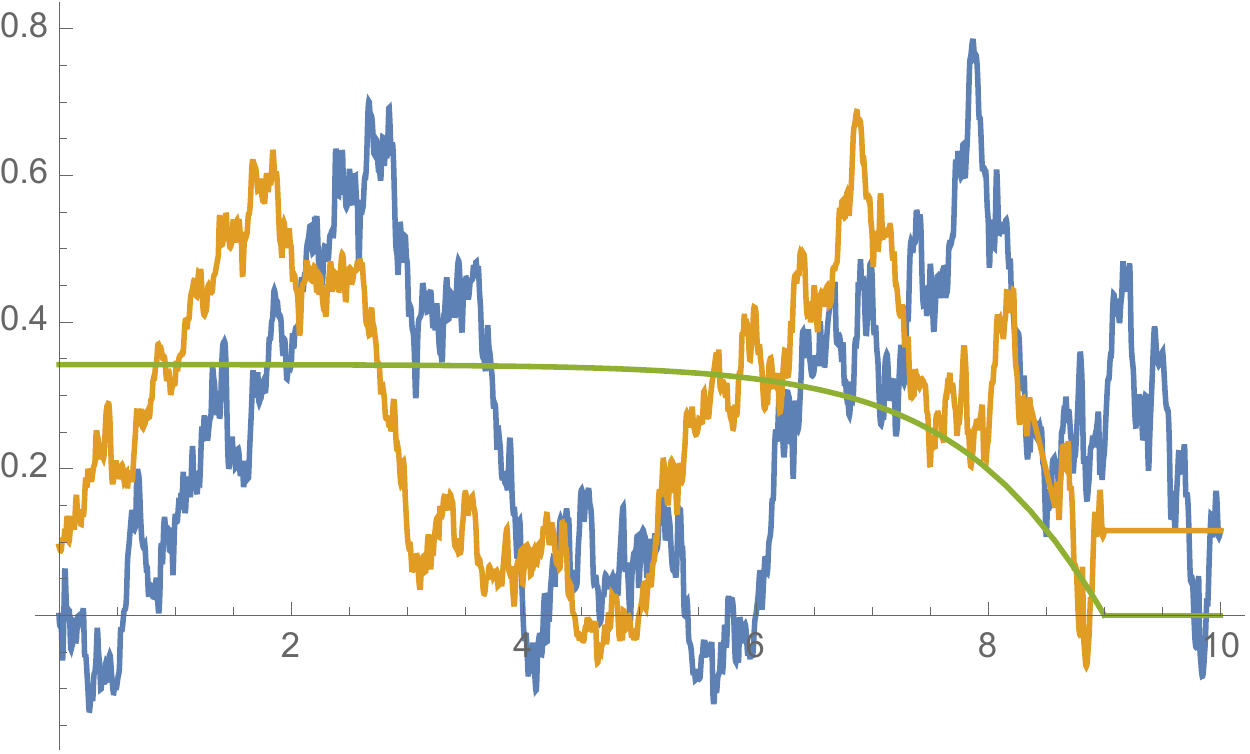}
\includegraphics[width=0.3\textwidth]{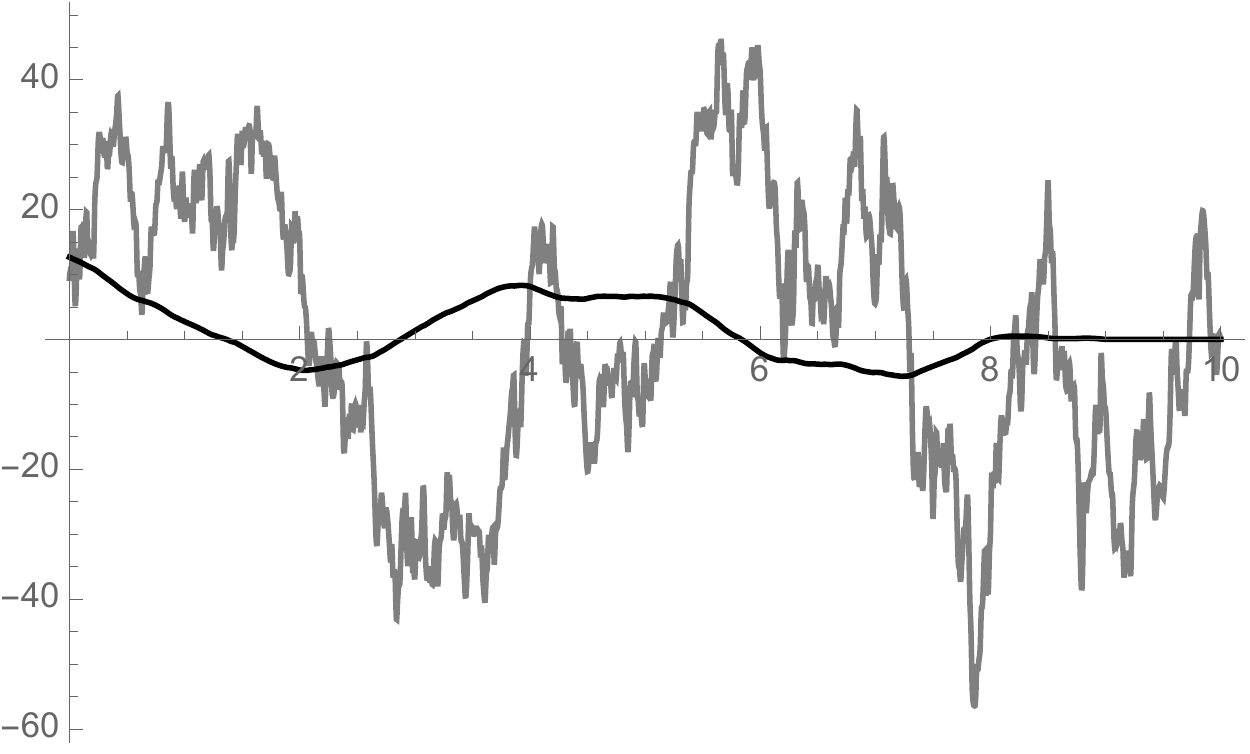}
\includegraphics[width=0.3\textwidth]{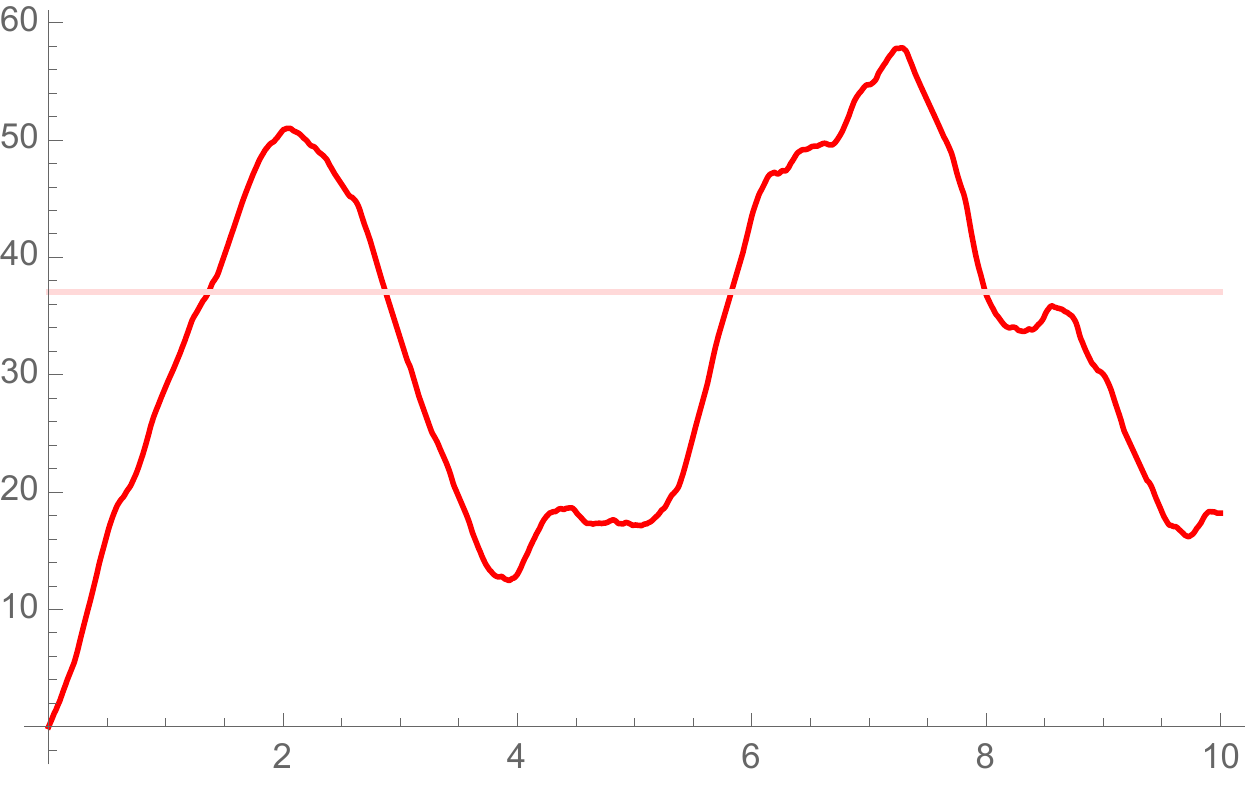}
\caption{The first of these illustrations shows an evolution of the stock price
  $S$ (blue), the corresponding average $S^\Delta$
  (orange) along with the underlying weight $\Upsilon^\Delta$; the second shows
  the resulting trading rates due to ``frontrunning'' (grey) and
  due to tracking the Merton portfolio (black); the third display
  shows the ensuing stock position $\Phi$ (red) together with the
  Merton ration $\mu/(\alpha \sigma^2$ (light red). Parameters where
  chosen as $s_0=0$, $\mu=.1$, $\sigma=.3$, $T=10$, $\Delta=1$,
  $\alpha=.03$, $\Phi_0=0$, $\Lambda = .01$.}
\end{figure}

The monetary value of being able to peek ahead by $\Delta$ is best described by
the certainty equivalent
  \begin{align}\label{eq:7}
c(\Delta) &= -\frac{1}{\alpha} \log\frac{ \max_{\phi \in \mathcal A^{\Delta}} \mathbb  E\left[-\exp\left(-\alpha
                                                V^{\Phi_0,\phi}_T\right)\right]}{ \max_{\phi \in \mathcal A^{0}} \mathbb  E\left[-\exp\left(-\alpha
                                                V^{\Phi_0,\phi}_T\right)\right]}\\
    &=\frac{1}{2\alpha}  \int_{0}^T
      \frac{(s\wedge\Delta)\rho}{1+(s\wedge\Delta)
      \sqrt{\rho}\tanh\left(\sqrt{\rho}(T-s)\right)}ds
  \end{align}
determined by comparing the utility attainable for an informed investor
(with admissible strategy set $\mathcal A^\Delta$) and an uninformed one (who
is confined to strategies from the smaller class $\mathcal A^0$)\footnote{The
      integral in~\eqref{eq:7} can be computed explicitly, but the resulting
  formulae turn out to be not more informative than the above integral
  and are therefore omitted.}.
Interestingly, the certainty equivalent does not depend on
the stock's risk premium $\mu$, but is determined by the
risk/liquidity ratio $\rho=\alpha \sigma^2/\Lambda$, the investor's
time horizon $T$ and the time units $\Delta$ she can look
ahead. Except for a period of length $\Delta$ and with a lot of time
to go, it accrues at about the rate
$\Delta \rho/(2(1+\Delta \sqrt{\rho}))$ which increases with $\Delta$
to the upper bound $\sqrt{\rho}/2$, revealing again the curb frictions
put on the earning potential of even extreme information advantages.

The proof of Theorem~\ref{thm2.1} is carried out in the next
section. It is obtained by solving the dual problem to
\begin{align}
  \intertext{$\qquad$Minimze}\label{eq:dualproblem}
  \mathbb E_{\mathbb Q}\left[\Phi_0(S_T-S_0) +\frac{1}{2\Lambda}\int_{0}^T\left|
\mathbb E_{\mathbb Q}\left[S_T\middle|\mathcal
  G^{\Delta}_t\right]-S_t\right|^2 dt\right]+\frac{1}{\alpha}
  \mathbb E_{\mathbb Q}\left[\log \frac{d\mathbb Q}{d\mathbb P}\right]
  \intertext{$\qquad$over $\mathbb Q \sim \mathbb P$ with finite relative entropy
$\mathbb E_{\mathbb Q}\left[\log \frac{d\mathbb Q}{d\mathbb
    P}\right]<\infty$.}
\end{align}
The corresponding duality theory holds true beyond the Brownian
framework specified here and is developed in a self-contained manner
in the Appendix~\ref{appendix} as a second key contribution of our
paper.



\section{Proof of Theorem~\ref{thm2.1}}\label{sec:proof}
Let us first note that it suffices to treat the case
\begin{align}\label{eq:wlog}
  S=W\text{, i.e., without loss of generality $s_0=0$, $\sigma=1$, $\mu=0$.}
\end{align}
Indeed, by passing from $\alpha$, $\Lambda$, $\mu$ to, respectively,
$\alpha'=\alpha\sigma$, $\Lambda'=\Lambda/\sigma$, $\mu'=\mu/\sigma$,
the utility with $\sigma'=1$ obtained from a given strategy will
coincide with the one obtained from this strategy under the original
parameters. Moreover, rewriting the expected utility under
$\mathbb P' \sim \mathbb P$ with density
$$\frac{d{\hat {\mathbb P'}}}{d\mathbb P}{|\mathcal F^W_T}:=\exp\left(-\mu'
  W_T-\frac{1}{2}\mu'^2 T\right),$$ under which
$W'_t=W_t+\mu' t$, $t \geq 0$, is a driftless Brownian motion, the expected
utilities under $\mathbb P$ coincide, up to the factor $\exp\left(\frac{1}{2}\mu'^2 T\right)$, with those under $\mathbb P'$ if we start with
$\Phi_0'=\Phi_0-{\mu}/({\alpha \sigma^2})$ rather than $\Phi_0$
risky assets.

The proof of Theorem \ref{thm2.1} will be accomplished via the dual
problem whose properties are summarized in the following proposition
which is an immediate consequence of the general duality results
presented in Appendix~\ref{appendix}.
\begin{prop}\label{prop2.1}
Denoting by $\mathcal Q$ the set of all probability measures $\mathbb Q\sim\mathbb P$ with finite
entropy
$$\mathbb E_{\mathbb Q}\left[\log\left(\frac{d\mathbb Q}{d\mathbb P}\right) \right]<\infty$$
relative to $\mathbb P$, we have
\begin{align}
\max_{\phi\in\mathcal A}&\left\{-\frac{1}{\alpha}\log\mathbb E\left[\exp\left(-\alpha V^{\Phi_0,\phi}_T\right)\right]\right\}\\
&=\min_{\mathbb Q\in\mathcal Q}\mathbb E_{\mathbb Q}\left[\Phi_0(S_T-S_0)+\frac{1}{\alpha}\log\left(\frac{d\mathbb Q}{d\mathbb P}\right) +\frac{1}{2\Lambda}\int_{0}^T\left|
\mathbb E_{\mathbb Q}(S_T|\mathcal G^{\Delta}_t)-S_t\right|^2 dt\right].\label{eq:1}
\end{align}
Furthermore, the minimizer $\hat{\mathbb Q}$ for the dual
problem is unique  and yields via
\begin{equation}\label{portfolio}
\hat\phi_t:=\frac{\mathbb E_{\hat{\mathbb Q}}\left[S_T|\mathcal G^{\Delta}_t\right]-S_t}{\Lambda}, \quad t\in [0,T],
\end{equation}
the unique optimal portfolio for the primal problem.
\end{prop}
\begin{proof}Follows from Proposition~\ref{propapp} below with the
  choice $\mathcal{G}_{t}:=\mathcal{G}_{t}^{\Delta}$ after
noting that $S_{t}/\sqrt{t}$ is standard Gaussian and so
$$
\sup_{t\in [0,T]}\mathbb E[\exp(a S_{t}^{2})]\leq \mathbb E[\exp(a S_{T}^{2})]<\infty,
$$
clearly holds for some small enough $a>0$.
\end{proof}

In order to solve the utility maximization problem it therefore
suffices to find the minimizer $\hat{\mathbb Q}$ of the dual problem
and work out the conditional expectation in~\eqref{portfolio}. This is
the path we will follow for the rest of this section. In a first step
we derive a particularly convenient representation for the target
functional of our dual problem:

\begin{lem}\label{lem:1}
  The dual infimum in~\eqref{eq:1} coincides with the one taken over
  all $\mathbb Q \in \mathcal Q$ whose densities take the form
  \begin{align}
\frac{d\mathbb Q}{d\mathbb P} = \exp\left(-\int_0^T\theta_t
    dW_t-\frac{1}{2}\int_0^T\theta_t^2dt\right) \label{eq:standarddensity}
  \end{align}
  for some bounded and adapted $\theta$ changing values only at finitely
  many deterministic times. For such $\mathbb Q$ the induced value~\eqref{eq:dualproblem} for
  the dual problem can be written as
  \begin{align}
    &\mathbb E_{\mathbb Q}\left[\Phi_0(S_T-S_0)+\frac{1}{\alpha}\log\left(\frac{d\mathbb Q}{d\mathbb P}\right) +\frac{1}{2\Lambda}\int_{0}^T\left|
\mathbb E_{\mathbb Q}\left[S_T\middle|\mathcal G^{\Delta}_t\right]-S_t\right|^2
      dt\right]\\
    =&-\Phi_0\int_{0}^T a_t dt+\frac{1}{2\alpha}\int_{0}^T a^2_t
      dt+\frac{1}{2\Lambda}\int_{0}^T \left(\int_t^T a_udu\right)^2 dt\\
&+\int_{0}^T  \mathbb E_{\mathbb Q}\left[\frac{1}{2\alpha}\int_{s}^{T}
                                                                             l^2_{t,s} dt+
\frac{1}{2\Lambda}
\int_{s}^{T}\left(\int_t^T l_{u,s}du\right)^2dt
+\frac{s\wedge\Delta}{2\Lambda}\left(1-\int_s^Tl_{u,s}du\right)^2\right]ds
  \end{align}
  where, for $t\in[0,T]$, $a_t$ and $l_{t,.}$ are
  determined by the It\^o-representations
  \begin{align}\label{eq:3}
     \theta_t = a_t + \int_0^t l_{t,s} dW^{\mathbb Q}_s
  \end{align}
  with respect to the $\mathbb Q$-Brownian motion $W^{\mathbb
    Q}_s=W_s+\int_0^s \theta_r dr$, $s \geq 0$.
\end{lem}
\begin{proof}
  For any $\mathbb Q \in \mathcal Q$ the martingale representation
  property of Brownian motion gives us a predictable $\theta$ with
  $\mathbb E_{\mathbb Q}[\log(d\mathbb Q/d\mathbb P)]=\mathbb
  E_{\mathbb Q}[\int_0^T\theta^2_sds]/2<\infty$ such that the density
  $d\mathbb Q/d\mathbb P$ takes the
  form~\eqref{eq:standarddensity}. Using this density to rewrite the
  dual target functional as an expectation under $\mathbb P$, we can
  follow standard density arguments to see that the infimum over
  $\mathbb Q \in \mathcal Q$ can be realized by considering the
  $\mathbb Q$ induced via~\eqref{eq:standarddensity} by simple
  $\theta$ as described in the lemma's formulation. As a consequence,
  the It\^o representations of $\theta_t$ in~\eqref{eq:3} can be
  chosen in such a way that the resulting $(a_t,l_{t,.})$ are also
  measurable in $t$: in fact they only change when $\theta$ changes
  its value, i.e., at finitely many deterministic times.  This joint
  measurability will allow us below to freely apply Fubini’s theorem.

  Let us rewrite the dual target functional in terms of $a$ and $l$. In
  terms of $\theta$ and the $\mathbb Q$-Brownian motion
  $W^{\mathbb Q}$, it reads
  \begin{align}\label{eq:firstrepdual}
\mathbb E_{\mathbb Q}&\left[\Phi_0(S_T-S_0)+\frac{1}{\alpha}\log\left(\frac{d\mathbb Q}{d\mathbb P}\right) +\frac{1}{2\Lambda}\int_{0}^T\left(
\mathbb E_{\mathbb Q}\left[S_T\middle|\mathcal G^{\Delta}_t\right]-S_t\right)^2 dt\right]\\
&=\mathbb E_{\mathbb Q}\Bigg[-\Phi_0\int_{0}^T \theta_t dt+\frac{1}{2\alpha}\int_{0}^T \theta^2_u du\\
&\qquad\qquad+\frac{1}{2\Lambda}\int_{0}^T
\left(W^{\mathbb Q}_{(t+\Delta)\wedge T}-W^{\mathbb Q}_t-\mathbb E_{\mathbb Q}\left[\int_{t}^T\theta_u du\middle|\mathcal G^{\Delta}_t\right]\right)^2 dt\Bigg].
 \end{align}
From It\^o's isometry and Fubini's theorem we obtain
\begin{align}\label{3.3}
\mathbb E_{\mathbb Q}\left[\int_{0}^T \theta^2_u du\right]=\int_{0}^T a^2_t dt+\int_{0}^T\int_{s}^{T}\mathbb E_{\mathbb Q}\left[l^2_{t,s} \right]dt\, ds.
\end{align}
Again by Fubini's theorem it follows that
\begin{align}
\mathbb E_{\mathbb Q}\left[\int_{t}^T \theta_udu\middle|\mathcal G^{\Delta}_t\right]&=
\int_{t}^T a_u du+\mathbb E_{\mathbb Q}\left[\int_{0}^T\int_{t\vee s}^T l_{u,s}du\, dW^{\mathbb Q}_s\middle|\mathcal G^{\Delta}_t\right]\\
&=\int_{t}^T a_u du+\int_{0}^{(t+\Delta)\wedge T} \int_{t\vee s}^T l_{u,s} du\, dW^{\mathbb Q}_s
\end{align}
for any $t\in [0,T]$, where the last equality follows from the
martingale property of stochastic integrals. Thus, another application of It\^o's isometry yields
\begin{align}
\mathbb E_{\mathbb Q}&\left[\left(W^{\mathbb Q}_{(t+\Delta)\wedge T}-W^{\mathbb Q}_t-\mathbb E_{\mathbb Q}\left[\int_{t}^T \theta_u du\middle|\mathcal G_t^\Delta\right]\right)^2 \right]\\
&=\left(\int_{t}^T a_u du\right)^2+ \mathbb E_{\mathbb Q}\left[\int_{0}^t \left(\int_{t}^T l_{u,s} du\right)^2ds \right]\\
&\qquad
+\mathbb E_{\mathbb Q}\left[\int_{t}^{(t+\Delta)\wedge T} \left(1-\int_{s}^T l_{u,s} du\right)^2ds \right].
\end{align}
Plugging this together with~\eqref{3.3} into~\eqref{eq:firstrepdual}
and using Fubini's theorem then provides us with the claimed formula
for our dual target value:
\begin{align}
\mathbb E_{\mathbb Q}&\left[-\Phi_0(S_T-S_0)+\frac{1}{\alpha}\log\left(\frac{d\mathbb Q}{d\mathbb P}\right) +\frac{1}{2\Lambda}\int_{0}^T\left(
\mathbb E_{\mathbb Q}\left[S_T\middle|\mathcal G^{\Delta}_t\right]-S_t\right)^2 dt\right]\\
=&-\Phi_0\int_{0}^T a_t dt+\frac{1}{2\alpha}\int_{0}^T a^2_t
      dt+\frac{1}{2\Lambda}\int_{0}^T \left(\int_t^T a_udu\right)^2 dt\\
&+\int_{0}^T \mathbb E_{\mathbb Q}\left[\frac{1}{2\alpha}\int_{s}^{T}
                                                                             l^2_{t,s} dt+
\frac{1}{2\Lambda}
\int_{s}^{T}\left(\int_t^T l_{u,s}du\right)^2dt\right.\\&\qquad\qquad\qquad\qquad\qquad\qquad\qquad\left.
+\frac{s\wedge\Delta}{2\Lambda}\left(1-\int_s^Tl_{u,s}du\right)^2\right]ds.
\end{align}
\end{proof}
The crucial point of the above representation is that by taking the
minimum separately over $a$ and over $l_{.,s}$ for each $s \in [0,T]$
we obtain \emph{deterministic} variational problems that can be solved
explicitly (see the next Lemma~\ref{lem:2}) \emph{and} this
deterministic minimum yields a lower-bound for the dual target value
that, using some Gaussian process theory, will ultimately be shown to
actually coincide with it (see Lemma~\ref{lem:3} below).

\begin{lem}\label{lem:2}
   Recall our notation $\rho=\alpha\sigma^2/\Lambda=\alpha/\Lambda$
   (because $\sigma=1$; cf.\ \eqref{eq:wlog}).
  \begin{enumerate}
  \item[(i)] The minimum of the functional
    \begin{align}
-\Phi_0\int_{0}^T a_t dt+\frac{1}{2\alpha}\int_{0}^T a^2_t
      dt+\frac{1}{2\Lambda}\int_{0}^T \left(\int_t^T a_udu\right)^2 dt
    \end{align}
    over $a \in L^2([0,T],dt)$ is attained for $\hat{a}\Phi_0$ where
    \begin{align}\label{eq:4}
      \hat{a}_t = \frac{\alpha
      \cosh(\sqrt{\rho}(T-t))}{\cosh(\sqrt{\rho} T)}, \quad t \in [0,T].
    \end{align}
    The resulting minimum value is $-\hat{A}_T\Phi_0^2$ where
    \begin{align}\label{eq:hatA}
      \hat{A}_T=\Lambda\sqrt{\rho}\tanh(\sqrt{\rho} T)/2.
    \end{align}
 \item[(ii)] For any $s \in [0,T]$, the minimum of the functional
    \begin{align}
\frac{1}{2\alpha}\int_{s}^{T}
                                                                             l^2_{t} dt+
\frac{1}{2\Lambda}
\int_{s}^{T}\left(\int_t^T l_{u}du\right)^2dt
+\frac{s\wedge\Delta}{2\Lambda}\left(1-\int_s^Tl_{u}du\right)^2
    \end{align}
   over $l \in L^2([s,T],dt)$ is attained at
    \begin{align}\label{eq:5}
       \hat{l}_{t,s} =  \frac{ \rho (s \wedge \Delta)
      \cosh(\sqrt{\rho}(T-t))}{\cosh(\sqrt{\rho}(T-s))+ \sqrt{\rho}
      (s \wedge \Delta) \sinh(\sqrt{\rho}(T-s))}, \quad t\in[s,T].
    \end{align}
    The corresponding minimum value is
    \begin{align}\label{eq:hatL}
      \hat{L}_s=\frac{1}{2\Lambda}\frac{s \wedge  \Delta}{1+(s\wedge \Delta)\sqrt{\rho} \tanh(\sqrt{\rho}(T-s))}.
     \end{align}
    \end{enumerate}
  \end{lem}
\begin{proof}
  We start with (ii).  The uniqueness follows from strict convexity of
  the functional to be minimized over $l \in L^2([s,T],dt)$. To write
  this as a standard variational problem, put
  $H(u,v):=\frac{1}{2\Lambda}u^2+ \frac{1}{2\alpha} v^2$ for
  $ u,v\in\mathbb R$, reparametrize $l$ via $g(t) = \int_t^T l_u du$,
  $t \in [s,T]$, and
  consider, for any $s \in [0,T]$ and any $\Theta\in\mathbb R$, the
  problem to minimize $\int_{s}^T H(g_t,\dot{g}_t) dt$ over
  $g \in C^1[s,T]$ subject to the constraints $g(s)=\Theta$, $g(T)=0$.

The optimization problem is convex and so it has a unique solution which has to satisfy the
Euler–Lagrange equation (for details see Section~1 in \cite{GF:63})
$$\frac{d}{dt}\frac{\partial H}{\partial \dot{g}}=\frac{\partial H}{\partial{g}}.$$
Thus, the optimizer is the unique solution of the linear ODE
$$\ddot{g}=\rho g, \ \ g(s)=\Theta, \ \ g(T)=0,$$
namely
$$g^{\Theta,s}(t):=\frac{\Theta\sinh\left(\sqrt \rho(T-t)\right)}
{\sinh\left(\sqrt{\rho}(T-s)\right)}, \ \ t\in [s,T].$$

Next, observe that for the function $g(t):=\int_{t}^T i_u du$, $t\in [s,T]$ we have
$\dot{g}=-l$ where $\dot{g}$ is the weak derivative of $g$, and so,
$$\frac{1}{2\alpha}\int_{s}^T \psi^2_t dt+\frac{1}{2\Lambda}\int_{s}^T\left(\int_{t}^{T}\psi_u du\right)^2dt=\int_{s}^T H(g_t,\dot{g}_t) dt.$$
Thus, from simple density arguments (needed since $g$ is not necessarily smooth) we obtain that
\begin{align}
\inf_{l\in L^2 ([s,T],dt)}&\left\{\frac{1}{2\alpha}\int_{s}^{T}l^2_{t} dt+
\frac{1}{2\Lambda}
\int_{s}^{T}\left(\int_t^T l_{u}du\right)^2dt
+\frac{s\wedge\Delta}{2\Lambda}\left(1-\int_s^T l_{u}du\right)^2\right\}
\\
&=\inf_{\Theta\in\mathbb R}\left\{\int_{s}^T H(g^{\Theta,s}_t,\dot{g}^{\Theta,s}_t) dt+\frac{s\wedge\Delta}{2\Lambda}(1-\Theta)^2\right\}\\
&=\frac{1}{2\Lambda}\inf_{\Theta\in\mathbb R}
\left\{\frac {\coth\left(\sqrt{\rho}(T-s)\right)}{\sqrt{\rho}}\Theta^2+(s\wedge\Delta) (1-\Theta)^2\right\}
\end{align}
where the last equality follows from simple computations.

Finally, the minimum of the above quadratic pattern (in $\Theta$) is atained at
$$\Theta^{*}=-\frac{ (s\wedge\Delta)\sqrt{\rho}}{\coth\left(\sqrt{\rho}(T-s)\right)+(s\wedge\Delta)\sqrt{\rho}}.$$
This gives \eqref{eq:5}--\eqref{eq:hatL}.

The proof of~(i) is almost the same as the of~(ii), but slightly simpler.
Observe that
\begin{align}
\inf_{a\in L^2([0,T],dt)}&\left\{-\Phi_0\int_{0}^T a_t dt+\frac{1}{2\alpha}\int_{0}^T a^2_t
      dt+\frac{1}{2\Lambda}\int_{0}^T \left(\int_t^T a_udu\right)^2 dt\right\}\\
&=\inf_{\Theta\in\mathbb R}\left\{-\Phi_0 \Theta+\int_{0}^T H(g^{\Theta,0}_t,\dot{g}^{\Theta,0}_t) dt\right\}\\
&=\inf_{\Theta\in\mathbb R}\left\{-\Phi_0 \Theta+\frac {\coth\left(\sqrt{\rho}T\right)}{2\sqrt{\rho}\Lambda}\Theta^2\right\}.
\end{align}
The minimum of the above quadratic pattern (in $\Theta$) is attained in
$$\tilde{\Theta}^{*}=\Phi_0\sqrt{\rho}\Lambda \tanh\left(\sqrt{\rho}T\right).$$
This gives \eqref{eq:4}--\eqref{eq:hatA}.
\end{proof}

The previous two lemmas suggest a way to construct a candidate for the
solution to the dual problem: Find $\hat{\mathbb Q} \sim \mathbb P$
whose density is given by
 \begin{align} \label{eq:2}
\frac{d\hat{\mathbb Q}}{d\mathbb P} = \exp\left(-\int_0^T\hat\theta_t
    dW_t-\frac{1}{2}\int_0^T\hat\theta_t^2dt\right)
  \end{align}
with
\begin{align}
  \label{eq:hattheta}
  \hat{\theta}_s =\hat{a}_s\Phi_0+ \int_0^s\hat{l}_{s,r}
  d\hat{W}^{\hat{\mathbb Q}}_r, \quad s \in [0,T].
\end{align}
For the associated Brownian motion
$\hat{W}=W^{\hat{\mathbb Q}}=W+\int_0^. \hat{\theta}_r dr$ this implies the
Volterra-type integral equation
\begin{align}\label{eq:Volterra}
  W_t+\int_0^t \hat{a}_s  \Phi_0 ds  = \hat{W}_t-\int_0^t \int_0^s\hat{l}_{s,r}
  d\hat{W}_rds, \quad t \in [0,T].
\end{align}
Integral equations of this type occur in~\cite{H:68,HH:93}; see also \cite{Cheridito:01}. By
considering $W+\int_0^. \hat{a}_r \Phi_0 dr$ as a Brownian motion with
respect to some probability measure which is equivalent to
$\mathbb P$, we can apply the results from Section~6.4 in \cite{HH:93}
(in particular see Theorem 6.3 and its proof there). We obtain that
\eqref{eq:Volterra} has a unique solution given by
\begin{align}
  \label{eq:hatW}
\hat{W}_t& = W_t+\Phi_0\int_0^t \hat{a}_s ds
-\int_{0}^t\int_{0}^s \hat\kkernel_{s,r} \left(dW_r+\Phi_0\hat{a}_r dr\right) ds \nonumber\\
&=W_t-\int_{0}^t\int_{0}^s \hat\kkernel_{s,r} dW_r ds+\Phi_0\left(\int_{0}^t\hat{a}_s ds-\int_{0}^t\int_{0}^s \hat\kkernel_{s,r}\hat a_r dr ds\right)
\end{align}
where $\hat{\kkernel}$
is the associated resolvent kernel characterized by the equation
\begin{align}
   \label{eq:6}
\hat{\kkernel}_{t,s}+\hat{l}_{t,s}=\int_{s}^t
   \hat{l}_{t,u}\hat{\kkernel}_{u,s}du, \quad 0 \leq s \leq t \leq T.
\end{align}
Moreover, $\hat{W}$ is a Brownian motion with respect to $\hat{\mathbb Q}$ which is well defined by~\eqref{eq:2}.

As our $\hat{l}$ is separable multiplicatively, \eqref{eq:6}
can be reduced to a linear ODE from which we compute the explicit
solution
\begin{align}\label{eq:hatkappa}
 \hat{\kkernel}_{t,s}=-\exp\left(\int_s^t
  \hat{l}_{u,u}du\right)\hat{l}_{t,s}, \quad 0 \leq s \leq t \leq T .
\end{align}

We are now in a position to solve the dual problem:
\begin{lem}\label{lem:3}
The dual infimum~\eqref{eq:1} is attained by $\hat{\mathbb Q} \sim
\mathbb P$ with density
\begin{align}
  \label{eq:9}
  \frac{d\hat{\mathbb Q}}{d\mathbb P} = \exp\left(-\int_0^T
  \hat{\theta}_tdW_t-\frac{1}{2}\int_0^T \hat{\theta}_t^2dt\right)
\end{align}
for $\hat{\theta}$ constructed in~\eqref{eq:hattheta} with $\hat{W}$ as
given by~\eqref{eq:hatW}; this $\hat{W}$ coincides with the
$\hat{\mathbb Q}$-Brownian motion induced by the $\mathbb P$-Brownian
motion $W$ via Girsanov's theorem. The value of the dual problem is
\begin{align}\label{eq:18}
-\frac{\Lambda\Phi^2_0 \sqrt{\rho}}{2\coth\left(\sqrt{\rho}T\right)}+\int_{0}^T \frac{1}{2\Lambda}\frac{(s\wedge\Delta)}{1+(s\wedge\Delta) \sqrt{\rho}\tanh\left(\sqrt{\rho}(T-s)\right)}ds.
  \end{align}
\end{lem}
\begin{proof}
  The construction of $\hat{\mathbb Q}$, $\hat{W}$ and $\hat{\theta}$
  has already been established by the preceding discussion. It is
  readily checked that $\hat{\mathbb Q}$ has finite entropy relative
  to $\mathbb P$ and so $\hat{\mathbb Q} \in \mathcal Q$. Note that
  $\hat{W}$ and $W$ generate the same filtration because
  of~\eqref{eq:Volterra} and~\eqref{eq:hatW} and so we can follow the
  same reasoning as in the proof of Lemma~3.2 to obtain its
  representation for the dual target functional also for
  $\hat{\mathbb Q}$. Recalling the minimizing properties of $\hat{a}$ and
  $\hat{l}_{.,s}$, $s \in [0,T]$, it then follows that
  $\hat{\mathbb Q}$ solves the dual problem with value~\eqref{eq:18}.
\end{proof}
By~\eqref{eq:1} the above value~\eqref{eq:18} of the dual problem
already yields the claimed value~\eqref{eq:primalvalue} for our primal
utility maximization problem.  For the completion of the proof of
Theorem~\ref{thm2.1} it therefore remains to work out the optimal
turnover policy $\hat{\phi}$. Due to its dual
description~\eqref{portfolio}, it suffices to compute
$\mathbb E_{\hat{\mathbb Q}}\left[S_T|\mathcal G^{\Delta}_t\right]=
\mathbb E_{\hat{\mathbb Q}}\left[W_T|\mathcal
  F_{t+\Delta}\right]$. Recalling the Volterra-type
equation~\eqref{eq:Volterra} and using Fubini's theorem we can write
\begin{align}
  \label{eq:10}
     W_T&= \hat{W}_T- \int_0^T \left(\hat{a}_u\Phi_0+ \int_0^u\hat{l}_{u,s}
          d\hat{W}_s\right)du \\
  &= \int_0^T \left(1-\int_s^T
    \hat{l}_{u,s}du\right)d\hat{W}_s-\int_0^T\hat{a}_u du\Phi_0.
\end{align}
Thus, for any $t \in [0,T]$, we find
\begin{align}
  \label{eq:11}
  &\mathbb E_{\hat{\mathbb Q}}\left[S_T\middle|\mathcal G_t^\Delta\right] \\
  &=
\int_0^{(t+\Delta)  \wedge T} \left(1-\int_s^T \hat{l}_{u,s}du\right)d\hat{W}_s-
  \int_0^T\hat{a}_u du\Phi_0 \\
&= \int_0^{(t+\Delta)  \wedge T} \left(1-\int_s^T
                                                                \hat{l}_{u,s}du\right)\left(dW_s-\int_{0}^s \hat\kkernel_{s,r}
                                                                d W_r ds\right)\\
 &\qquad +
  \Phi_0 \left(\int_0^{{(t+\Delta)
  \wedge T}}\left(1-\int_s^T \hat{l}_{u,s}du\right)\left(\hat{a}_sds-\int_0^s\hat{\kkernel}_{s,r}\hat{a}_rdr\,ds\right)-\int_0^T\hat{a}_u du\right)
\end{align}
where in the second step we used~\eqref{eq:hatW} to get an expression
in terms of the original input to our problem $W$ rather than
$\hat{W}$. The structure of this expression suggests to consider for
$X=W$ and $X=\int_0^. \hat{a}_s ds$ the integral operator
\begin{align}
  \mathcal I_{t}^{T}(X) := \int_0^{(t+\Delta)\wedge T} \left(1-\int_s^T
                                                                \hat{l}^T_{u,s}du\right)\left(dX_s-\int_{0}^s \hat\kkernel^T_{s,r}
                                                                d X_r ds\right)
\end{align}
for continuous paths $X$. Notice that the $dX$-integrals can be
defined through integration by parts which reveals in particular that
$\mathcal I^T_t(X)$ depends continuously on $X$; notice also that we
used the notation $\hat{l}^T$ and $\hat{\kkernel}^T$ in lieu of $l$
and $\kkernel$ to emphasize that these kernels depend on the time
horizon $T$. In conjunction with~\eqref{portfolio} and $S_t=W_t$, this
provides us with a (somewhat) explicit `open loop' expression of the
optimal turnover policy:
\begin{align} \label{eq:openloop}
  \hat{\phi}_t = \frac{1}{\Lambda}
  \left(\mathcal{I}_{t}^{T}\left(W\right)-W_t
 +\Phi_0\left(\mathcal{I}_{t}^{T}\left(\int_0^.\hat{a}^T_sds\right)-\int_0^T\hat{a}^T_udu\right)\right),
\end{align}
where, again, $\hat{a}^T$ is used to recall that $\hat{a}$
of~\eqref{eq:4} depends on $T$.

To establish the policy's more informative feedback description given in
Theorem~\ref{thm2.1}, we note next that dynamic programming holds for
our problem:

\begin{lem}\label{lem:policy}
  The optimal policy
  $\hat{\phi}$ of~\eqref{eq:openloop} can alternatively be described in the form
  \begin{align}\nonumber
    \hat{\phi}_t = \frac{1}{\Lambda}
  \Bigg(&\mathcal{I}_{0}^{T-t}\left(W_{t+.}-W_t\right)-W_t\\         \label{eq:13}
&\qquad+\hat{\Phi}_t\left(\mathcal{I}_{0}^{T-t}\left(\int_0^.\hat{a}^{T-t}_sds\right)-\int_0^{T-t}\hat{a}^{T-t}_udu\right)\Bigg)
  \end{align}
   where $\hat{\Phi}_t=\Phi_0+\int_0^t \hat{\phi}_sds$ for $t \in [0,T]$.
 \end{lem}
\begin{proof}
  The righthand side of~\eqref{eq:openloop} gives us for each time
  horizon $T$ a continuous functional
  $\Psi^T:\mathbb R\times C[0,T]\rightarrow C[0,T]$ such that for any
  initial stock position $\Phi_0 \in \mathbb R$ and any stock price
  evolution $W$, $\Psi^T(\Phi_0,W|_{[0,T]})$ is the correspondingly
  optimal strategy $\hat{\phi}$ for the utility maximization problem.

Assume by contradiction that the statement of our lemma does not hold.
Then, by continuity of sample paths of $\hat{\phi}$, there exists $t_0\in [0,T]$ such
that with positive probability $\hat{\phi}_{t_0}$ does not coincide with the
righthand side of~\eqref{eq:13}.  Now consider the strategy $\tilde{\phi}$
that coincides with
$\hat{\phi}$ up to time $t_0$ when it switches to 
$$\tilde\phi_t:=\Psi^{T-t_0}\left(\hat\Phi_{t_0},W_{.+t_0}-W_{t_0}\right)_{t-t_0},
\quad t \in [t_0,T].$$
For any strategy $\phi$, we can write the contribution over the interval $[t_0,T]$ to the resulting terminal wealth as
\begin{align}
  V^{\Phi_0,{\phi}}_T - V^{\Phi_0,{\phi}}_{t_0}
  =\hat\Phi_{t_0}(S_T-S_{t_0})+\int_{t_0}^T{\Phi}_t(S_T-S_t)dt-\frac{\Lambda}{2}\int_{t_0}^T{\phi}^2_tdt=:V_{[t_0,T]}^{\Phi_{t_0},,\phi},
\end{align}
where $\Phi_{t_0}:=\Phi_0+\int_0^{t_0}\phi_t dt$. Of course,
$\tilde{\Phi}_{t_0}:=\Phi_0+\int_0^{t_0}\tilde{\phi}_t
dt=\hat{\Phi}_{t_0}$. So, by the Markov property of Brownian motion
and choice of $\tilde{\phi}$ as the unique optimal policy as of time
$t_0$, this allows us to observe that
\begin{align}
  \mathbb{E}\left[-\exp\left(-\alpha V_{[t_0,T]}^{\tilde{\Phi}_{t_0},\tilde{\phi}}\right)\middle|\mathcal
                                                             G^\Delta_{t_0}\right]
 \geq  \mathbb{E}\left[-\exp\left(-\alpha V_{[t_0,T]}^{\hat{\Phi}_{t_0},\hat{\phi}}\right)\middle|\mathcal
                                                             G^\Delta_{t_0}\right],
\end{align}
with ``$>$'' holding on $\{\hat{\phi}_{t_0} \not= \tilde{\phi}_{t_0}\}$
(i.e.\ where~\eqref{eq:13} is violated) because continuity of
$\hat{\phi}$ and $\tilde{\phi}$ ensures that they will differ on an
open interval once they differ at all. Since by assumption this
happens with positive probability,  it follows for the unconditional expected
utility from $\tilde{\phi}$ that
\begin{align}
\mathbb{E} \left[-\exp(-\alpha
  V^{\Phi_0,\tilde{\phi}}_T)\right]&=\mathbb{E} \left[\exp(-\alpha V_{t_0}^{\Phi_0,\tilde{\phi}})  \mathbb{E}\left[-\exp\left(-V_{[t_0,T]}^{\tilde{\Phi}_{t_0},\tilde\phi}\right)\middle|\mathcal
                                                             G^\Delta_{t_0}\right]\right]\\&>\mathbb{E} \left[\exp(-\alpha V_{t_0}^{\Phi_0,\tilde{\phi}})  \mathbb{E}\left[-\exp\left(-V_{[t_0,T]}^{\hat{\Phi}_{t_0},\hat\phi}\right)\middle|\mathcal
                                                             G^\Delta_{t_0}\right]\right]\\&=\mathbb{E} \left[-\exp(-\alpha
  V^{\Phi_0,\hat{\phi}}_T)\right],\label{eq:26}  
\end{align}
contradicting the optimality of $\hat{\phi}$.
\end{proof}

 As a consequence of this dynamic programming result, it
 suffices to verify our feedback policy
 description~\eqref{eq:feedback} for time $t=0$:

 \begin{lem}\label{lem:policy1}
   The optimal initial turnover rate is
 \begin{align}\label{eq:17}
\hat{\phi}_0=
&\frac{1}{1
+\Delta\sqrt{\rho}\tanh(\sqrt{\rho}(T-\Delta)^{+})}\frac{S_{\Delta \wedge T}}{\Lambda}
\\&+\frac{\sqrt{\rho}}{\coth(\sqrt{\rho}(T-\Delta)^{+}) +\Delta\sqrt{\rho}}
\int_0^{\Delta \wedge T} \frac{S_{s}}{\Lambda}ds-\frac{S_0}{\Lambda}\nonumber\\
                          &+\frac{\sqrt{\rho}}{\coth(\sqrt{\rho}(T-\Delta)^{+}) +\Delta\sqrt{\rho}}
                            \left(\frac{\mu}{\alpha\sigma^2}-\Phi_0\right).\nonumber
  \end{align}
 \end{lem}
 \begin{proof}
   In view of~\eqref{eq:openloop}, we need to compute for $X=W$ and
   $X=\int_0^.\hat{a}_udu$ the operator
 \begin{align}
   \mathcal{I}_{0}^T\left(X\right) &=
                        \int_0^{\Delta \wedge T} \left(1-\int_s^T
                                                                \hat{l}_{t,s}dt\right)\left(dX_s-\int_{0}^s \hat\kkernel_{s,r}
                                                                d X_r
                        ds\right)\\\label{eq:15}
   &= \int_0^{\Delta \wedge T} \left(1-\int_s^T  \hat{l}_{t,s}dt\right)dX_s-I
 \end{align}
 with
 \begin{align}
  I&:=\int_0^{\Delta \wedge T} \left(1-\int_s^T\hat{l}_{t,s}dt\right)
                             \int_{0}^s \hat\kkernel_{s,r} d X_r ds\\  \label{eq:16}
                           &= \int_0^{\Delta \wedge T}
                             \left(\int_r^{\Delta\wedge T}
                                \hat{\kkernel}_{s,r}ds
                             -\int_r^T\int_r^{t\wedge \Delta}
                             \hat{l}_{t,s}\hat{\kkernel}_{s,r}ds \ dt
                             \right)  dX_r
 \end{align}
 where the last equality is due to Fubini's theorem. For
 $t \in [r,\Delta \wedge T]$, the kernel identity~\eqref{eq:6} shows
 that the second $ds$-integral in~\eqref{eq:16} gives
 $\hat{\kkernel}_{t,r}+\hat{l}_{t,r}$. For $t \in (\Delta \wedge T,T]$,
 we note that $\Delta<T$ and we let $n_t$ denote the numerator in the
 definition of $\hat{l}_{t,.}$ in~\eqref{eq:5} to write
 $\hat{l}_{t,r}=\hat{l}_{\Delta,r} n_t/n_{\Delta}$. It follows by another use of the kernel identity~\eqref{eq:6}
 that for such $t$ the second $ds$-integral above amounts to
 \begin{align}
   \label{eq:19}
   \int_r^{\Delta} \hat{l}_{t,s}\hat{\kkernel}_{s,r}ds & =
     \int_r^{\Delta} \hat{l}_{\Delta,s}\hat{\kkernel}_{s,r}ds
                                                       \frac{n_t}{n_{\Delta}}
                                                     = \left(\hat{\kkernel}_{\Delta,r}+\hat{l}_{\Delta,r}\right) \frac{n_t}{n_{\Delta }}
                                                     \\ & =\hat{l}_{t,r}-\exp\left(\int_r^{\Delta}\hat{l}_{u,u}du\right)\hat{l}_{t,r}
 \end{align}
 where we used~\eqref{eq:hatkappa} in the final step. Plugging all
 this into~\eqref{eq:16}, we see that the contribution to the
 $dt$-integral from $[r,\Delta \wedge T]$ is partially cancelled by
 the first $ds$-integral there, leaving us with
 \begin{align}
   \label{eq:20}
     I=\int_0^{\Delta \wedge T}\left(-\int_r^T
   \hat{l}_{t,r}dt+\int_{\Delta \wedge T}^T\exp\left(\int_r^{\Delta}\hat{l}_{u,u}du\right)\hat{l}_{t,r}dt\right)dX_r.
 \end{align}
 Inserting this into~\eqref{eq:15}, we see a cancellation of integrals
 over $\hat{l}$ and arrive at
 \begin{align}
   \mathcal{I}_{0}^T\left(X\right) &=
                        \int_0^{\Delta \wedge T} \left(1-\int_{\Delta
                        \wedge
                        T}^T\exp\left(\int_s^{\Delta}\hat{l}_{u,u}du\right)\hat{l}_{t,s}dt\right)dX_s\\\label{eq:21}
   &=\int_0^{\Delta \wedge T} \left(1-f_T(s)\right)dX_s
 \end{align}
 where in view of \eqref{eq:5} we have
  \begin{align}\label{eq:fT}
    f_{T}(s) := &\exp\left(\int_s^{\Delta \wedge T}
                         \frac{u\rho}{1+u\sqrt{\rho}\tanh(\sqrt{\rho}(T-u))}du\right) \\
    & \qquad \cdot \frac{s\sqrt{\rho}\sinh(\sqrt{\rho}(T-\Delta)^+)}{\cosh(\sqrt{\rho}(\tau-s)) +s\sqrt{\rho}\sinh(\sqrt{\rho}(T-s))}\\
    &=\frac{s\sqrt{\rho}\sinh(\sqrt{\rho}(T-\Delta)^+)}{\cosh(\sqrt{\rho}(T-\Delta)^{+}) +\Delta\sqrt{\rho}\sinh(\sqrt{\rho}(T-\Delta)^{+})}.
  \end{align}
 Now we apply~\eqref{eq:21} to $X=W$ and $X=\int_.^T\hat{a}_udu$ to
 rewrite the open loop description~\eqref{eq:openloop} of
 $\hat{\phi}_0$ as
 \begin{align}
   \hat{\phi}_0 =&
\frac{W_{\Delta \wedge T}}{\Lambda}    (1-f_T(\Delta \wedge T))
                            +\int_0^{\Delta \wedge T} \frac{W_{s}}{\Lambda}  f'_T(s)ds\\
                          &-
                            \Phi_0\left(
                            \int_{\Delta \wedge
                            T}^T\frac{\hat{a}_u}{\Lambda}du    (1-f_T(\Delta \wedge T))+\int_0^{\Delta\wedge T}\int_{s}^T\frac{\hat{a}_u}{\Lambda}du f'_T(s)ds\right).
 \end{align}
We conclude the claimed representation for the optimal policy~\eqref{eq:17} by inserting~\eqref{eq:fT} and
 \begin{align}
   \label{eq:23}
   \int_{s}^T\frac{\hat{a}_u}{\Lambda}du =
   \frac{\sqrt{\rho}\sinh(\sqrt{\rho}(T-s))}{\cosh(\sqrt{\rho}T)}, \quad s \in [0,T],
 \end{align}
in the above formula for $\hat{\phi}_0$. 

As a final step, we need to recall the simplifying steps from the
beginning of this chapter where we reduced everything to the case
$S=W$ underpinning our calculations so far. Reversing these steps then
leads to the formulae given in the present lemma which work for the
general case required in our main theorem.
\end{proof}

\appendix{}
\numberwithin{equation}{section}
\section{Duality}\label{appendix}

In this appendix we develop the duality theory for the utility
maximization problem~\eqref{problem} not just for the Bachelier model
discussed in the rest of the paper, but for any c\`adl\`ag price
process $S=(S_t)_{t \in [0,T]}$ on a filtered probability space
$(\Omega,\mathcal{F},\mathbb{P})$ equipped with the (completed,
right-continuous) filtration $(\mathcal{G}_{t})_{t\in [0,T]}$ to which
$S$ is adapted. Expectation of a real-valued random variable
$X$ with respect to some probability ${R}$ on $\mathcal{F}$ is denoted $\mathbb{E}_{R}[X]$ where
the index is dropped when $\mathbb{Q}=\mathbb{P}$. Sometimes we also use the shorthand notation
$\mathbb{E}_{{R},t}[X]$ (resp.\ $\mathbb{E}_{t}[X]$) instead of $\mathbb{E}_{R}[X\vert\mathcal{G}_{t}]$
(resp.\ $\mathbb{E}[X\vert\mathcal{G}_{t}]$).

Define the set of admissible strategies by
$$
\mathcal{A}:=\left\{(\phi_{t})_{t\in [0,T]}:\
\int_{0}^{T}\phi_{t}^{2}\, dt<\infty\mbox{ almost surely, } \phi\mbox{ is an optional process}
\right\}
$$
and for each $\phi\in\mathcal{A}$ define the corresponding portfolio value at time $T$ by
$$
V^{\Phi_{0},\phi}:=\Phi_{0}(S_{T}-S_{0})+\int_{0}^{T}(S_{T}-S_{t})\phi_{t}\, dt-
\frac{\Lambda}{2}\int_{0}^{T}\phi_{t}^{2}\, dt,
$$
where $\Lambda>0$ is a given constant characterizing the strength of price impact.

\begin{asm}\label{integrability} There is $a>0$ such that
$$
\sup_{t\in [0,T]}\mathbb{E}[\exp(a S_{t}^{2})]<\infty.
$$
\end{asm}

\begin{prop}\label{propapp} Let Assumption \ref{integrability} be in force.
Denoting by $\mathcal Q$ the set of all probability measures $\mathbb Q\sim\mathbb P$ with finite
entropy
$$\mathbb E_{\mathbb Q}\left[\log\left(\frac{d\mathbb Q}{d\mathbb P}\right) \right]<\infty$$
relative to $\mathbb P$, we have
\begin{align}
&\max_{\phi\in\mathcal A}\left\{-\frac{1}{\alpha}\log \mathbb{E}\left[\exp\left(-\alpha V^{\Phi_0,\phi}_T\right)\right]\right\}\nonumber\\
&=\inf_{\mathbb Q\in\mathcal Q} \mathbb{E}_{\mathbb Q}\left[\Phi_0(S_T-S_0)+\frac{1}{\alpha}\log\left(\frac{d\mathbb Q}{d\mathbb P}\right) +\frac{1}{2\Lambda}\int_{0}^T\left|
\mathbb{E}_{\mathbb Q}[S_T|\mathcal G_t]-S_t\right|^2 dt\right].\label{eq:1masik}
\end{align}
Furthermore, there is a unique minimizer $\hat{\mathbb Q}$ for the dual
problem and the process given by
\begin{equation}\label{portfoliomasik}
\hat{\phi}_t:=\frac{\mathbb{E}_{\hat{\mathbb Q}}\left[S_T|\mathcal G_t\right]-S_t}{\Lambda}, \quad t\in [0,T],
\end{equation}
is the unique optimal portfolio for the primal problem.
\end{prop}

We will prove Proposition \ref{propapp} at the end of this section, after suitable preparations.
In the rest of this section we assume $\alpha=1$, $\Lambda=2$, $\Phi_{0}=0$ for simplicity
and will write $V(\phi)$ instead of $V^{0,\phi}$.
The general case is only notationally more involved.

We first state the ``difficult'' direction of the superhedging theorem in the present
context: it provides a sufficient condition for a claim to be superhedged by
a suitable strategy.

\begin{thm}\label{superhedging}
Let $\mathbb{Q}\sim \mathbb{P}$ be a probability such that
\begin{equation}\label{hol}
\sup_{t\in [0,T]}\mathbb{E}_{\mathbb{Q}}\left[|S_{t}|\right]<\infty.{}
\end{equation}
Let $W$
be a real-valued random variable with $\mathbb{E}_{\mathbb{Q}}[|W|]<\infty$. If
\begin{equation}\label{igaz}
\mathbb{E}_{{R}}[W]\leq \frac{1}{4}\mathbb{E}_{{R}}\left[\int_{0}^{T}(\mathbb{E}_{{R},t}[S_{T}]-S_{t})^{2}\,dt\right]
\end{equation}
holds for all probabilities ${R}\ll \mathbb{Q}$ with bounded $d{R}/d\mathbb{Q}$
then there exists $\phi\in\mathcal{A}$ such that $V(\phi)\geq W$ almost surely.
\end{thm}
\begin{proof} It follows along the lines of the case treated in Theorem 3.9 of \cite{paolo},
the integrability condition \eqref{hol} being used in the arguments corresponding to those of page 2082 there.	
\end{proof}

Note that the convex conjugate of the function $u(x):=-\exp(-x)$, $x\in\mathbb{R}$ is
$$v(y):=\sup_{x\in\mathbb{R}}[u(x)-xy],\ y\geq 0.
$$
By simple calculations
$v(y)=y\ln y-y$, where the convention $0\ln 0=0$ is used.
The \emph{Fenchel inequality} $u(x)\leq v(y)+xy$ trivially holds for all $x\in\mathbb{R}$ and
$y\geq 0$.

Let $\mathcal{Z}$ denote the set of non-negative random variables $\xi$ such that
$\mathbb{E}[\xi]=1$. For each $\xi$, define the probability $R(\xi)\ll \mathbb{P}${}
by $R(\xi)(A):=\mathbb{E}[\xi 1_{A}]$. Let $\mathcal{Z}_{e}:=\{\xi\in\mathcal{Z}_{e}:\ \mathbb{E}[\xi |\ln\xi|]<\infty\}$.

In the rest of this section, Assumption \ref{integrability} will be in force.
Since we will follow a standard route, described e.g.\ in \cite{kabanov-stricker},
only the main steps of the proofs will be given.

\begin{lem}\label{moments}
For any family $\mathcal Z_0 \subset \mathcal Z$ with $\sup_{\xi \in
  \mathcal Z_0}\mathbb E[\xi \ln \xi]<\infty$, one has
$$
\sup_{\xi \in \mathcal Z_0, t\in [0,T]}\mathbb{E}[\xi S_{t}^{2}]<\infty.
$$	
\end{lem}
\begin{proof} Consider the conjugate Orlicz spaces corresponding to the Young
functions $\Phi(x)=e^{x}-x-1$ and $\Psi(x)=(1+x)\ln(1+x)-x$,
their respective norms being denoted by $||\cdot||_{\Phi}$, $||\cdot||_{\Psi}$, see
the Appendix of \cite{neveu} for definitions.
Assumption \ref{integrability} then implies $\sup_{t\in [0,T]}||S_{t}^{2}||_{\Phi}<\infty$.
Using Proposition A-2-2 of \cite{neveu}, we get that
\begin{equation}\label{usi}
\mathbb{E}[\xi S_{t}^{2}]\leq C ||\xi||_{\Psi}||S_{t}^{2}||_{\Phi}	
\end{equation}
with some constant $C$. The statement follows.
\end{proof}

\begin{lem}\label{finiteness} The functional
\begin{equation}\label{xixo}
\Xi(\xi):=\mathbb{E}[\xi\ln\xi]+\frac{1}{4}\mathbb{E}_{R(\xi)}\left[{}
\int_{0}^{T}(\mathbb{E}_{R(\xi),t}[S_{T}]-S_{t})^{2}\, dt\right]
\end{equation}
is finite and strictly convex on the convex set $\mathcal{Z}_{e}$.
\end{lem}
\begin{proof} Finiteness of the functional follows from Lemma
  \ref{moments}, convexity of the set $\mathcal{Z}_{e}$ is
  easy. Strict convexity of the first summand defining $\Xi$ is easy
  to see; convexity of the second part is somewhat more strenuous.
  The main ingredient is the following: for any random variable
  $X\in \cap_{\xi\in\mathcal{Z}_{e}}L^{1}(R(\xi))$ and sigma-algebra
  $\mathcal{H}$ the mapping
$$
\xi\in\mathcal{Z}_{e}\to \mathbb{E}_{R(\xi)}[(\mathbb{E}_{R(\xi)}[X\vert\mathcal{H}])^{2}]
$$
is convex. Indeed, for $s\in [0,1]$ and $\xi_{1},\xi_{2}\in\mathcal{Z}_{e}$
one should check that
\begin{eqnarray*}
& & \mathbb{E}_{R(s\xi_{1}+(1-s)\xi_{2})}[(\mathbb{E}_{R(s\xi_{1}+(1-s)\xi_{2})}[X\vert\mathcal{H}])^{2}]\\	
&=& \mathbb{E}\left[(s\xi_{1}+(1-s)\xi_{2})
\left(\frac{\mathbb{E}[(s\xi_{1}+(1-s)\xi_{2})X\vert\mathcal{H}]}{\mathbb{E}[s\xi_{1}+(1-s)\xi_{2}\vert\mathcal{H}]}\right)^{2}
\right]\\
&=& \mathbb{E}\left[\frac{\mathbb{E}^{2}[(s\xi_{1}+(1-s)\xi_{2})X\vert\mathcal{H}]}
{\mathbb{E}[s\xi_{1}+(1-s)\xi_{2}\vert\mathcal{H}]}\right]\\	
&\leq& \mathbb{E}\left[
\frac{s\mathbb{E}^{2}[\xi_{1}X\vert\mathcal{H}]}{\mathbb{E}[\xi_{1}\vert\mathcal{H}]}\right]+
\mathbb{E}\left[\frac{(1-s)\mathbb{E}^{2}[\xi_{2}X\vert\mathcal{H}]}{\mathbb{E}[\xi_{2}\vert\mathcal{H}]}\right],
\end{eqnarray*}
where, for simplicity, we assume $\xi_{1},\xi_{2}>0$ almost surely, the general case
being only notationally more difficult.
This is equivalent, by elementary calculations, to checking
\begin{eqnarray*}
& & 2\mathbb{E}[\xi_{1}X\vert\mathcal{H}]\mathbb{E}[\xi_{2}X\vert\mathcal{H}]\mathbb{E}[\xi_{1}\vert\mathcal{H}]\mathbb{E}[\xi_{2}\vert\mathcal{H}]
\\
&\leq&
\mathbb{E}^{2}[\xi_{1}X\vert\mathcal{H}]\mathbb{E}^{2}[\xi_{2}\vert\mathcal{H}]+
\mathbb{E}^{2}[\xi_{2}X\vert\mathcal{H}]\mathbb{E}^{2}[\xi_{1}\vert\mathcal{H}]
\end{eqnarray*}
which is clearly true. Now applying this observation with the choice $\mathcal{H}=\mathcal{G}_{t}$
and $X=S_{T}-S_{t}$ for each $t$ the statement follows easily.
\end{proof}

\begin{prop}\label{minimizer}
There exists a unique minimizer $\underline{\xi}$ of $\Xi$ on $\mathcal{Z}_{e}$
which is positive almost surely. We denote $\underline{J}:=\inf_{\xi\in\mathcal{Z}_{e}}\Xi(\xi)$.
\end{prop}
\begin{proof} Uniqueness of a minimizer is immediate from the strict
  convexity of $\Xi$. Let $\xi_{n}\in\mathcal{Z}_{e}$ be a minimizing
  sequence for $\Xi$. By the Koml\'os theorem, there exist
  C\'esaro-means of a subsequence (denoted by $\tilde{\xi}_{n}$)
  converging almost surely to some $\underline{\xi}$. As
  $\mathcal{Z}_{e}$ is convex,
  $\tilde{\xi}_{n}\in\mathcal{Z}_{e}$. Convexity of $\Xi$ implies that
\begin{equation}\label{mos}
\Xi(\tilde{\xi}_{n})\to\underline{J}	
\end{equation} still
holds. As $\sup_{n}\mathbb{E}[\tilde{\xi}_{n}|\ln\tilde{\xi}_{n}|]<\infty$ follows from \eqref{mos},
the de la Vall\'ee-Poussin criterion ensures that $\tilde{\xi}_{n}$ converges in $L^{1}$, too,
hence $\underline{\xi}\in\mathcal{Z}_{e}$.

We claim that
$\mathbb{E}_{R(\tilde{\xi}_{n}),t}[S_{T}]\to
\mathbb{E}_{R(\underline{\xi}),t}[S_{T}]$ in probability and observe
that, by Fatou's lemma and~\eqref{mos}, this entails
$\Xi(\underline{\xi})\leq \underline{J}$, establishing
$\underline{\xi}$ as a minimizer of $\Xi$ in $\mathcal Z_e$. Since
$\tilde{\xi}_{n}S_{T}\to\underline{\xi}S_{T}$ almost surely, the
claimed convergence will follow from the uniform integrability
condition
\begin{align}\label{eq:1015}
\sup_{n}\mathbb
  E\left[\tilde{\xi}_{n}|S_{T}|1_{\{\tilde{\xi}_{n}|S_{T}|\} \geq b}\right]
  \to 0 \text{ as } b \uparrow \infty.
\end{align}
For this we estimate
\begin{align}
  \label{eq:1012}
 \mathbb  E\left[\tilde{\xi}_{n}|S_{T}|
 1_{\{\tilde{\xi}_{n}|S_{T}|\geq b\}}
 \right] & \leq   \mathbb
           E\left[\tilde{\xi}_{n}|S_{T}|^2\right]^{1/2}\mathbb E\left[\tilde{\xi}_n
           1_{\{\tilde{\xi}_{n}|S_{T}|\geq b\}}
\right]^{1/2} 
\end{align}
The first of the above factors is bounded uniformly in $n$ due to
Lemma~\ref{moments}. We will get our assertion~\eqref{eq:1015} by
showing that the second factor vanishes uniformly in $n$ as
$b \uparrow\infty$. For this we estimate it further:
\begin{align}
  \label{eq:14}
  \mathbb E\left[\tilde{\xi}_n1_{\{\tilde{\xi}_{n}|S_{T}|\geq
  b\}}\right]
  &\leq   \mathbb
    E\left[\tilde{\xi}_n1_{\{\ln(\tilde{\xi}_{n}(|S_{T}|\vee 1))\geq
   \ln b\}}\right]
   \leq   \mathbb
    E\left[\tilde{\xi}_n\frac{\ln^+(\tilde{\xi}_{n}(|S_{T}|\vee 1))}{\ln
    b}\right]\\&\leq\frac{1}{\ln b}\left(\mathbb
    E\left[\tilde{\xi}_n\ln^+\tilde{\xi}_{n}\right]+\mathbb
    E\left[\tilde{\xi}_n\ln^+\left(|S_{T}|\vee 1\right)\right]\right).
\end{align}
The first of these last two expectations is bounded uniformly due
to~\eqref{mos}, the second because of this in conjunction with
Lemma~\ref{moments}, and we can conclude.

To prove positivity of $\underline{\xi}$, let us pick a strictly positive
$\xi_{0}\in\mathcal{Z}_{e}$ (for instance $\xi_{0}\equiv 1$) and
define
\begin{equation}\label{ef}
F_{s}:=\Xi(s\xi_{0}+(1-s)\underline{\xi}),\ s\in [0,1].	
\end{equation}
By optimality of $\underline{\xi}$ the right derivative $F_{0+}'$ is non-negative.
If we had $P(\underline{\xi}=0)>0$ then we would reach a contradiction just like in Proposition 3.1 of
\cite{kabanov-stricker}.	
\end{proof}

\begin{thm}\label{maximizer}
There exists a strategy
$\phi^{\dagger}$ such that $V(\phi^{\dagger})=\underline{J}-\ln\underline{\xi}$.
Moreover,
\begin{equation}\label{primalproblem}
\sup_{\phi\in\mathcal{A}}\mathbb{E}[-\exp(-V(\phi))]=\mathbb{E}[-\exp(-V(\phi^{\dagger}))].
\end{equation}
\end{thm}
\begin{proof}
Notice that $V(\phi)\leq Q:=\frac{1}{4}\int_{0}^{T}(S_{T}-S_{t})^{2}\, dt$ for all $\phi$,{}
and $Q$ is $R(\xi)$-integrable for all $\xi\in\mathcal{Z}_{e}$, by the arguments
of Lemma \ref{moments}.
Let us fix now an arbitrary $\phi\in\mathcal{A}$ and $\xi\in\mathcal{Z}_{e}$.
For each $s>0$, we apply the Fenchel inequality and Fubini's theorem, remembering the
integrability of the quantity $Q$:
\begin{eqnarray}
& & \mathbb{E}[-\exp(-V(\phi))]\nonumber \\
&\leq& \mathbb{E}[s\xi\ln (s\xi)]-\mathbb{E}[s\xi]+\mathbb{E}\left[s\xi
\left(-\int_{0}^{T}S_{t}\phi_{t}\, dt-
\int_{0}^{T}\phi_{t}^{2}\, dt+\int_{0}^{T}S_{T}\phi_{t}\, dt\right)\right]\nonumber \\
&=& \mathbb{E}[s\xi\ln (s\xi)]-s+
s\int_{0}^{T}\mathbb{E}_{R(\xi)}\left[-S_{t}\phi_{t}-
\phi_{t}^{2}+S_{T}\phi_{t}\right]\, dt\nonumber \\
&\leq& \mathbb{E}[s\xi\ln (s\xi)]-s+
s\mathbb{E}\left[\xi\int_{0}^{T}\left(-S_{t}\phi_{t}-\phi_{t}^{2}+
\mathbb{E}_{R(\xi),t}[S_{T}]\phi_{t}\right)dt
\right]\nonumber
\\
&\leq& s\ln s + s\mathbb{E}[\xi\ln \xi]-s+
\frac{s}{4}\mathbb{E}\left[\xi\int_{0}^{T}(\mathbb{E}_{R(\xi),t}[S_{T}]-S_{t})^{2}\, dt\right].
\end{eqnarray}

Optimizing in $s$ we arrive at
\begin{eqnarray*}
\mathbb{E}[-\exp(-V(\phi))] \leq -\exp\left(-\mathbb{E}[\xi\ln \xi]-
\frac{1}{4}\mathbb{E}\left[\xi\int_{0}^{T}(\mathbb{E}_{R(\xi),t}[S_{T}]-S_{t})^{2}\, dt\right]\right).	
\end{eqnarray*}


Choosing $\xi:=\underline{\xi}$,
\begin{equation}\label{csirmi}
\mathbb{E}[-\exp(-V(\phi))]\leq -e^{-\underline{J}}.	
\end{equation}

Now we will prove that a suitable strategy $\phi^{\dagger}$ attains the bound given in \eqref{csirmi}.
Let $\xi\in\mathcal{Z}$ be such that $\xi\leq C_{0}\underline{\xi}$ for
some $C_{0}>0$. Since $\underline{\xi}|\ln\underline{\xi}|\leq C_{0}\underline{\xi}
(|\ln\underline{\xi}|+|\ln C_{0}|)$,
we have, in fact, $\xi\in\mathcal{Z}_{e}$.
Consider the function \begin{equation}
F_{s}:=\Xi(s\xi+(1-s)\underline{\xi}),\ s\in [0,1].	
\end{equation}
Since $\underline{\xi}$ is the minimizer, $F_{s}\geq F_{0}$ for $s\in [0,1]$
so the right-hand derivative satisfies
\begin{equation}\label{copacabana}
F_{0+}'\geq 0.	
\end{equation}
To simplify notation,
we will write $X_{t}:=S_{T}-S_{t}$ henceforth.
Let us calculate the latter derivative now. It equals
$$
\mathbb{E}[(\underline{\xi}-{\xi})\ln(\underline{\xi})]+
\frac{1}{4}\int_{0}^{T}\mathbb{E}\left[\frac{\mathbb{E}_{t}[\underline{\xi}]2\mathbb{E}_{t}[\underline{\xi}X_{t}]
\mathbb{E}_{t}[(\xi-\underline{\xi})X_{t}]-\mathbb{E}_{t}[\xi-\underline{\xi}]\mathbb{E}^{2}_{t}[\underline{\xi}X_{t}]}
{\mathbb{E}_{t}^{2}[\underline{\xi}]}\right] dt.
$$
Grouping the terms inside the integral that do not contain $\xi$ we obtain
$$
-\mathbb{E}_{R(\underline{\xi})}\left[\frac{\mathbb{E}^{2}_{t}[\underline{\xi}X_{t}]}{\mathbb{E}^{2}_{t}[\underline{\xi}]}\right].
$$
The rest of the integrand is
$$
\mathbb{E}\left[{}
2\frac{\mathbb{E}_{t}[\underline{\xi}X_{t}]\mathbb{E}_{t}[\xi X_{t}]}{\mathbb{E}_{t}[\underline{\xi}]}\right]+
\frac{-\mathbb{E}_{t}[\xi]\mathbb{E}^{2}_{t}[\underline{\xi}X_{t}]}{\mathbb{E}^{2}_{t}[\underline{\xi}]}.
$$
The first term of the latter expression can be rewritten and estimated by Cauchy's inequality:
$$
2\mathbb{E}\left[\frac{\mathbb{E}_{t}[\underline{\xi}X_{t}]\mathbb{E}_{t}[\xi X_{t}]\mathbb{E}_{t}[\xi]}{\mathbb{E}_{t}[\underline{\xi}]\mathbb{E}_{t}[\xi]}\right]
\leq \mathbb{E}\left[\frac{\mathbb{E}_{t}^{2}[\underline{\xi}X_{t}]\mathbb{E}_{t}[\xi]}{\mathbb{E}_{t}^{2}[\underline{\xi}]}\right]
+ \mathbb{E}\left[\frac{\mathbb{E}_{t}^{2}[{\xi}X_{t}]\mathbb{E}_{t}[\xi]}{\mathbb{E}_{t}^{2}[{\xi}]}\right].
$$
Taking all terms into consideration, condition \eqref{copacabana} eventually implies
\begin{eqnarray}\nonumber & &
\mathbb{E}[-\xi\ln(\underline{\xi})]\\
&\leq&
\mathbb{E}[-\underline{\xi}\ln(\underline{\xi})]
-\frac{1}{4}\mathbb{E}_{R(\underline{\xi})}\left[\int_{0}^{T}\frac{\mathbb{E}^{2}_{t}[\underline{\xi}X_{t}]}{\mathbb{E}^{2}_{t}[\underline{\xi}]}\, dt\right]
+\frac{1}{4}\mathbb{E}_{R({\xi})}\left[\int_{0}^{T}\frac{\mathbb{E}^{2}_{t}[{\xi}X_{t}]}{\mathbb{E}^{2}_{t}[{\xi}]}\, dt\right].{}
\label{haromcsillag}
\end{eqnarray}
We now use Theorem \ref{superhedging} with the choice $Q=R(\underline{\xi})$ and
$R=R(\xi)$. We obtain from \eqref{haromcsillag} and \eqref{igaz} that
$-\ln\underline{\xi}+\underline{J}\leq V(\phi^{\dagger})$ for some $\phi^{\dagger}\in\mathcal{A}$.
Since
$$
\mathbb{E}[-\exp(-V(\phi^{\dagger}))]\geq \mathbb{E}[-\exp(\ln\underline{\xi}-\underline{J})]=-e^{-\underline{J}}\mathbb{E}[\underline{\xi}]=
-e^{-\underline{J}},
$$
$\phi^{\dagger}$ is indeed an optimal strategy. Note also that,
by \eqref{csirmi}, the above inequality must be an a.s.\ equality so
$\ln\underline{\xi}=\underline{J}-V(\phi^{\dagger})$.
\end{proof}

\begin{cor}\label{biden}
The strategy
$$
\hat{\phi}_{t}:=\frac{\mathbb{E}_{R(\underline{\xi}),t}[S_{T}]-S_{t}}{2},\ t\in [0,T]
$$
is optimal for the problem \eqref{primalproblem}.	
\end{cor}
\begin{proof}
Indeed, the arguments of the previous theorem show that,
for all $\phi\in\mathcal{A}$, $\xi\in\mathcal{Z}_{e}$ and $s>0$,
\begin{eqnarray*}& &
\mathbb{E}[-\exp(-V(\phi))]\\
&\leq&
s\ln s + s\mathbb{E}[\xi\ln \xi]-s+
\frac{s}{4}\mathbb{E}\left[\xi\int_{0}^{T}(\mathbb{E}_{R(\xi),t}[S_{T}]-S_{t})^{2}\right]\\
&=& s\ln s + s\mathbb{E}[\xi\ln \xi]-s+
s\mathbb{E}\left[\xi V(\hat{\phi})\right]
\\ &\leq& -\exp\left(-\mathbb{E}[\xi\ln \xi]-
\frac{1}{4}\mathbb{E}\left[\xi\int_{0}^{T}(\mathbb{E}_{R(\xi),t}[S_{T}]-S_{t})^{2}\right]\right),	
\end{eqnarray*}
with equality for $\phi=\phi^{\dagger}$, $\xi=\underline{\xi}$ and for a suitable $s=s^{*}$.
This implies that $\hat{\phi}$ is an optimal strategy.
\end{proof}

\begin{proof}[Proof of Proposition \ref{propapp}]
Proposition \ref{minimizer} and Theorem \ref{maximizer} establish that the optimal
portfolio wealth $V(\phi^{\dagger})$ for the (primal) utility maximization problem
\eqref{primalproblem} can be found by first finding the (dual) minimizer $\underline{\xi}$ of the functional
\eqref{xixo} and then taking $\phi^{\dagger}$ satisfying $V(\phi^{\dagger})=\underline{J}-\ln\underline{\xi}$.
Finally, from the strict concavity of the map $\phi\rightarrow -\exp(-V(\phi))$ we conclude that
the strategy
$\hat{\phi}$ which is given in Corollary~\ref{biden} is the unique optimal strategy.
\end{proof}


\bibliography{finance}{}
\bibliographystyle{abbrv}
\end{document}